\documentclass[runningheads,journal]{llncs}
\usepackage[cmex10]{amsmath}
\usepackage{graphics}
\usepackage{algorithm}
\usepackage{fullpage}
\usepackage{url}
\urldef{\mail}\path|tigran.tonoyan@unige.ch|

\begin{document}

\mainmatter  

\title{On Wireless Scheduling Using the Mean Power Assignment}
\author{Tigran Tonoyan  \thanks {Research partially founded by FRONTS 215270.}}

\authorrunning{Tigran Tonoyan} 
\institute{TCS Sensor Lab\\
Centre Universitaire d'Informatique\\
 Route de Drize 7, 1227 Carouge, Geneva, Switzerland\\
 \mail\\
 \url{http://tcs.unige.ch}}

\maketitle              

\begin{abstract}
In this paper the problem of scheduling with power control in wireless networks is studied: given a set of communication requests, one needs to assign the powers of the network nodes, and schedule the transmissions so that they can be done in a minimum time, taking into account the signal interference of concurrently transmitting nodes. The signal interference is modeled by SINR constraints. Approximation algorithms are given for this problem, which use the \emph{mean power assignment}. The problem of schduling with fixed mean power assignment is also considered, and approximation guarantees are proven.
\end{abstract}

\section{Introduction}

One of the basic issues in wireless networks is that concurrent transmissions may cause interference.  This makes it necessary to schedule the transmission requests in the network, so that the interference of the concurrent transmissions on each transmission is small enough. On the other hand, the length of the schedule should be small, to make the delay caused by the interference as small as possible. We are interested in the problem of scheduling with power control, i.e. we choose the power levels of the nodes and then schedule the set of communication requests with respect to the chosen power settings.

The scheduling problem has been studied in several communication models. It has been shown that the results obtained in different models differ essentially.
One of the factors on which the scheduling problem crucially depends is the model of interference. Wireless networks have often been modeled as graphs. The nodes of this communication graph represent the physical devices, two nodes being connected by an edge if and only if the respective devices are within mutual transmission range. In this graph-theoretic model a node is assumed to receive a message correctly if and only if no other node in close physical proximity transmits at the same time. Clearly, the graph-theoretic model fails to capture the accumulative nature of actual radio signals. If the power levels of the nodes are chosen properly, then a node may successfully receive a message in spite of being in the transmission range of other simultaneous transmitters.

In contrast, during recent years there has been a significant amount of research done considering the problem of scheduling in models of wireless networks which are more realistic (and more efficient, see~\cite{moscibroda}) than graph-theoretic models.  The standard model is the \emph{SINR (signal-to-interference-plus-noise-ratio)} model. The SINR model reflects  physical reality more accurately and is therefore often simply called \emph{the physical model}.

More formally the problem of scheduling with power control (or simply \emph{PC-scheduling problem}) can be described as follows. Given is an arbitrary set of links, each a sender-receiver pair of wireless nodes. We seek an assignment of powers to the sender nodes and a partition of the link-set into a minimum number of subsets or \emph{slots}, so that the links in each slot satisfy the SINR-constraints w.r.t. the chosen power assignment. The problem is considered in two communication models: \emph{the directed model of communication}, where the communication between two nodes is one-directional (i.e. in one session only the sender node sends packets and the receiver just receives), and the \emph{bidirectional model of communication}, where both 
nodes in a link may be transmitting, which implies stronger constraints. 
We are trying to design algorithms that result in efficient schedules. 

As for the power assignments, we are particularly interested in schedules using so-called \emph{oblivious power assignments}, which depend only on the length of the given link. These power assignments are important in distributed networks, where each node has to chose its own power level based on a local information. So it is desirable to find short  schedules using these power assignments, or find  out how much worse can perform such power assignments in comparison to the optimal power assignment.

\textit{Related Work.}
The body of algorithmic work on the scheduling problem is mostly on graph-based models. The inefficiency of graph-based protocols has been shown theoretically as well as experimentally (see~\cite{gronkvist} and~\cite{moscibroda} for example). The problem of scheduling for networks arbitrarily located on the Euclidean plane, as opposed to the network instances with nodes uniformly scattered on some area of plane, is considered in~\cite{beginning1},~\cite{beginning2}. They design algorithms for assigning power of the nodes and scheduling a given set of links, but no approximation guarantees are proven. There is a series of papers considering the problem of scheduling and related problem of \emph{capacity} for given powers (the problem of capacity is to find a maximal subset of links, which can transmit concurrently). The case of uniform power assignment is considered in~\cite{santienq},~\cite{goussevskaya},~\cite{gousshall},~\cite{andrews}, whith a constant factor approximation algorithm designed in~\cite{hallwat}. In~\cite{beginning3},~\cite{fanghaenel} and~\cite{tonoyanlin} scheduling with linear power assignment is considered, obtaining a constant factor centralized algorithm and a distributed algorithm with a good approximation guarantee. In~\cite{kessvok} a $O(\log^2{n})$-approximation randomized distributed algorithms are designed for the family of \emph{sub-linear} power assignments, and in~\cite{hallmit} constant factor approximation algorithms (centralized) are designed for the capacity problem for the same class of power assignments. There is also a considerable effort towards finding power assignments, which would yield better results for scheduling and capacity problems (the problem of PC-scheduling). In~\cite{oblivious} the bidirectional version of PC-scheduling problem is considered, and it is shown that the mean power assignment yields a poly-logarithmic (in the number of links $n$) approximation factor. 
In~\cite{halldorsson},~\cite{tonoyanpc},~\cite{halldorsson1} it is shown that when using the mean power assignment, one can get a $O(\log{n})$-approximation for PC-scheduling in the bidirectional model, and a $O(\log{n}\log{\log{\Lambda}})$-approximation in the directed model, where $\Lambda$ is the ratio between the longest and the shortest link-lengths. In~\cite{kesselheim} a constant factor approximation algorithm is given for capacity maximization problem (with power control), which uses non-local power assignments. In fact it has been shown~\cite{oblivious}, that in the directed model for each oblivious power assignment $P$ there is a network instance, which is SINR-feasible with some power assignment, but yields an unefficient schedule using $P$, but constructed network instances are quite unnatural. 
A variant of PC-scheduling problem, modeling also \emph{multicast} transmissions, is considered in~\cite{erlebach}, and a $\log{\Lambda}$-approximation algorithm is proposed, which uses uniform power assignment. 
Some basic structural properties of a network in the SINR model, such as geometric properties of \emph{reception zones} of a set of nodes, and the use of those in point query algorithms are considered in~\cite{diagrams} and~\cite{diagrams1}.

\textit{About This Paper.}
This paper is based on~\cite{tonoyanpc}, in which it was shown that the results of~\cite{halldorsson} needed correction, and  alternative algorithms were proposed for scheduling using the \emph{mean power assignment}. Here a modification of the scheduling algorithm from~\cite{halldorsson} is proposed, and it is shown that it can be used to obtain relatively short schedules for \emph{independent sets of links}. Then, using some facts from geometric graph theory, it is shown that there is a constant factor approximation algorithm, which partitions a given set of links into independent subsets. Combining these results approximation algorithms are obtained for PC-scheduling problem: a $O(\log{n})$-approximation algorithm for the bidirectional model of communication, and a $O(\log^2{n}\log{\log{\Lambda}})$-approximation algorithm for the directed model of communication (this has been improved to $O(\log{n}\log{\log{\Lambda}})$ in~\cite{halldorsson1}). In both cases the mean power assignment is used. We also show that these algorithms are $O(\log{n})$-approximation algorithms for scheduling problem, when the power assignment is fixed to mean power.

\section{Problem Formulation and Preliminaries}
\subsection{SINR Constraint and Scheduling Problem}

The links in the network are represented by the set $L=\{1,2,\dots,n\}$, where each link $v\in L$ represents a communication request between a sender node $s_v$ and a receiver node $r_v$. The nodes are located on points in a metric space (we will often treat nodes as points) with distance function $d$. The \emph{asymmetric distance} $d_{vw}$ from a link $v$ to a link $w$ is defined in two ways, depending on which \emph{communication model} is adopted: $d_{vw} = d(s_v,r_w)$ in the \emph{directed communication model} and $d_{vw} =\min\{d(s_v,r_w),d(s_v,s_w),d(r_v,r_w),d(r_v,s_w)\}$ in the \emph{bidirectional model.}
Note that in the latter case $d_{vw}=d_{wv}$ (i.e. the distance is actually symmetrical), but in the former case for some pairs $v$,$w$ it can be $d_{vw}\neq d_{wv}$.

The \emph{length} of a link $v$ is  $l_v=d(s_v,r_v)$.

There is a power assignment $P:L\rightarrow R_+$, which assignes a positive number $P_v$ to each link $v$. This value determines the power of transmission of a transmitting node in $v$. In the directed model only the sender node is transmitting, so the power assignment means assigning powers to the sender nodes. In the bidirectional model the communication is bilateral, so both sender and receiver nodes of a link are assigned the same power.

We adopt the \emph{path loss radio propagation} model for the reception of signals, where   the signal received from a node $x$ of the link $v$ at some node $y$ is $P_v/d(x,y)^\alpha$, where $\alpha>2$ denotes the \emph{path loss exponent}. We adopt the \emph{physical interference model}, where a communication $v$ is done successfully if and only if the following condition holds:
\begin{equation}\label{E:sinr}
\frac{P_v/l_v^\alpha}{\sum_{w\in S\setminus\{v\}}{P_w/d_{wv}^\alpha}+N}\geq\beta,
\end{equation}
where $N$ denotes the ambient noise, $S$ is the set of concurrently scheduled links in the same \emph{slot}, and $\beta\geq 1$ denotes the minimum SINR(signal-to-interference-plus-noise-ratio) required for the transmission to be successfully done. We say that $S$ is \emph{SINR-feasible} if (\ref{E:sinr}) holds for each link in $S$. 

In the problem of \emph{scheduling with power control (PC-scheduling)}  given the set $L$ of links, one needs to choose a power assignment, and split $L$ into  SINR-feasible subsets (slots) with respect to the chosen power assignment, such that the number of slots is the minimum. The collection of such subsets is called \emph{schedule}, and the number of slots in a schedule is called \emph{the length} of the schedule. In the problem of \emph{scheduling with given powers} given the set $L$ and a power assignment, one needs to schedule $L$ into minimum number of slots with respect to the given power assignment. 

Note that each of these problems can be stated for both directed and bidirectional model. If for some statement we don't explicitly mention the model, then it is stated for both models.

\subsection{Fading Metrics}

We consider \emph{doubling metric spaces}~\cite{heinonen} in this paper. Such a metric space has a characteristic number, which is called \emph{doubling dimension}. We will use the property of doubling metric spaces, which is, each ball of radius $r$ contains at most $C\cdot (r/r')^m$ disjoint balls of a smaller radius $r'$, where $C$ is a constant, and $m$ is the doubling dimension.
It is known that the $k$-dimensional Euclidean space is a doubling metric space with doubling dimension $k$ (see~\cite{heinonen}).

We assume that the path loss exponent $\alpha$ is greater than the doubling dimension of the metric space. The pair of a doubling space and the path loss exponent greater than the dimension is called a \emph{fading metric}.

\subsection{Affectance and $p$-signal Sets}

At first we assume $N=0$ (i.e. there is no ambient noise), $\beta=1$, and strict inequality in (\ref{E:sinr}). We will show that thanks to Theorem~\ref{T:robustness} these assumptions do not have essential effect on the results.
With this assumptions it is convenient to consider the \emph{affectance} of a link $v$ caused by a set of links $S$, which is the inverse of SINR:
\[
a_S(v)=\sum_{w\in S\setminus\{v\}}{\frac{P_w/d_{wv}^\alpha}{P_v/l_v^\alpha}}=\sum_{w\in S\setminus\{v\}}{\frac{P_w}{P_v}\cdot \frac{l_v^\alpha}{d_{wv}^\alpha}}
\]
An important property of affectance is that it is additive, i.e. if there are two disjoint sets $S_1$ and $S_2$, then $a_{S_1\cup S_2}(v)=a_{S_1}(v)+a_{S_2}(v)$.

A \emph{$p$-signal} set or schedule is one where the affectance of any link is less than $1/p$. Note that a set is SINR-feasible if and only if it is a 1-signal set. We will call 1-signal schedule a \emph{SINR-feasible} schedule.

The following result demonstrates the robustness of schedules against small changes on the right side of SINR constraint.
Suppose the power assignment of the nodes is given.
\begin{theorem}\label{T:robustness}\cite{hallwat}
There is a polynomial-time algorithm that takes a $p$-signal schedule and refines into a $p'$-signal schedule, for $p'>p$, increasing the number of slots by a factor of at most $\lceil 2p'/p\rceil^2$.
\end{theorem}

The algorithm described in Theorem~\ref{T:robustness} works for both communication models.

\subsection{Independent Sets of Links}

 We call two links $v$ and $w$  \emph{$q$-independent} w.r.t. power assignment $\{P_v\}$, if 
 \[
 a_w(v)<1/q^\alpha\mbox{ and }a_v(w)<1/q^\alpha.
 \]

We are particularly interested in the \emph{mean} power assignment, which is given by assigning to a node of each link $v$ a power $P_v=cl_v^{\alpha/2}$, where $c>0$ is a constant. It is easy to check, that two links $v$ and $w$ are $q$-independent w.r.t. the mean powers if and only if $d_{vw} > q \sqrt{l_w l_v}$  and $d_{wv} > q \sqrt{l_w l_v}.$
In the bidirectional model $d_{wv}$ and $d_{vw}$ are equal, so the links $v$ and $w$ are $q$-independent with the mean powers if and only if $d_{vw} > q \sqrt{l_w l_v}.$

We call two links $v$ and $w$ $q$-independent, if the following inequality holds:
\[
d_{vw}d_{wv} > q^2 l_w l_v.
\]
 Note that for the bidirectional model two links are $q$-independent if and only if they are $q$-independent with the mean power assignment. 

A set $S$ of links is a \emph{$q$-independent} set if each pair of links in $S$ is $q$-independent.

The following lemma immediately follows from the definition of $q$-independence.
\begin{lemma}\label{L:signalslot}
A set of links that belong to the same $q^\alpha$-signal slot in some schedule, is $q$-independent.
\end{lemma}

We say that a set of links is \emph{nearly equilength}, if the lengths of any pair of links in the set differ not more than two times. 

The following theorem from~\cite{halldorsson} shows that each $q$-independent set $S$ of nearly equilength links in a fading metric is a $\Omega(q^\alpha)$-signal slot when the uniform powers are used, i.e. all nodes have the same power $P$, for some $P>0$.
\begin{theorem}\label{T:uniforms}\cite{halldorsson}
Let $L$ be a $q$-independent set of nearly equilength links in a fading metric. Then $L$ is a $\Omega(q^\alpha)$-signal set when the powers are uniform.
\end{theorem}

\section{Scheduling $q$-independent Sets}\label{S:scheduleindependent}

As it is shown in~\cite{tonoyanpc}, there is a flaw in the proofs of~\cite{halldorsson}, so their results stated for general metrics are still unproven. Here we show that their algorithm can be modified to work for scheduling $q$-independent sets. 

We assume that the network nodes are placed in a fading metric. We need the following definitions to state the scheduling algorithm.
 
A set $S$ of links is called \emph{well-separated}, if for each two links $v,w\in S,$ the we have $\max\{l_v/l_w,l_w/l_v\}\notin(2,n^2)$.

 Two links $v$ and $w$ are said to be \emph{$\tau$-close} under the mean power assignments if $\max\{a_v(w),a_w(v)\}\geq\tau$, i.e. at least one affects the other one more than by $\tau$.

A set of links $S\subseteq L$ is called \emph{$p$-bounded} for $p>0$, if for each link $v\in L$, there are at most $p$ links $w\in S$, such that $n^2l_v\leq l_w$ and $w$ and $v$ are $\displaystyle \frac{1}{2n}$-close.

 Let $q\geq 1$ be a constant. Consider a $q$-independent subset $Q$ of $L$. We describe a procedure, which, if $Q$ is $p$-bounded for some $p>0$, schedules $Q$ into $O(p\log{n})$ slots using the mean power assignment.  The pseudocode is presented in Algorithm~\ref{A:indep}.

\renewcommand {\labelenumi}{\arabic{enumi}.}
\renewcommand{\labelenumii}{\arabic{enumi}.\arabic{enumii}}
\begin{algorithm}\caption{Scheduling independent sets of links.}\label{A:indep}
\begin{enumerate}
\item {Input: a $q$-independent $p$-bounded set $Q$, for some $p>0$ and $q\geq1$}
\item {Let $Q=\cup_i{Q_i}$, where $Q_i=\{t\in Q|l_t\in [2^{i-1}l_{min},2^il_{min})\}$}
\item {Assign $B_i=\cup_j{Q_{i+j\cdot 2\log{n}}}$, for $i=1,2,\dots,2\log{n}$}
\item {Schedule each $B_i=\cup_j{K_j}$, where $K_j=Q_{i+j\cdot2\log{n}}$, the following way}
\begin{enumerate}
\item {Using the algorithm from Theorem~\ref{T:robustness} transform each $K_j$ into an $f$-signal schedule $\Sigma_j=\{S_j^s\}_{s=1}^{k_j}$ with  $f=2^{\alpha/2+1}$}
\item {$s\leftarrow 1$}
\item {Assign $S\leftarrow \cup_j{S_j^s}$: if for some $j$, $k_j<s$, then we take $S_j^s= \emptyset$}\label{I:return}
\item {Sort $S$ in the non-increasing order of link lengths: $l_1\geq l_2\geq\dots l_{|S|}$}
\item {$T_s^r\leftarrow \emptyset, r=1,2,\dots,p+1$}
\item {For $k=1,2,\dots,|S|$ do: find a $T_s^r$ not containing links $u$ with $l_u>n^2l_k$ which are $1/(2n)$-close to $k$, and assign $T_s^r\leftarrow T_s^r\cup\{k\}$ }
\item {$s\leftarrow s+1$: if $s\leq \max{k_j}$, then go to step 4.3, otherwise the schedule for $B_i$ is $\{T_s^r|T_s^r\neq\emptyset\}$}
\end{enumerate}
\item{Output the union of the schedules of all $B_i$}
\end{enumerate}
\end{algorithm}

The algorithm splits the input set into a logarithmic number of well-separated subsets $B_i$, then schedules each $B_i$ separately. First $B_i$ is split into \emph{maximal} equilength subsets $Q_j$ ($l_{min}$ in the algorithm is the minimal link-length). Then each $Q_j$ is scheduled into a constant number of slots with the mean power assignment, using Theorem~\ref{T:uniforms}. To schedule $B_i$, the algorithm takes the union of the first slots of the schedules for all $Q_j$ (which are contained in $B_i$), and schedules them into $p+1$ slots, using the fact that $Q$ is $p$-bounded. Hence we get a schedule with $O(p)$ slots for each $B_i$, and a schedule with $O(p\log{n})$ slots for $Q$. The  correctness of the algorithm is proven in the following theorem.
\begin{theorem}\label{T:scheduleindependent}
Let $Q=\{1,2,\dots,k\}$ be a $q$-inde\-pen\-dent $p$-bounded subset of $L$ for $q\geq 1$. Then Algorithm~\ref{A:indep} schedules $Q$ into $O(p\log{n})$ slots w.r.t. the mean power assignment.
\end{theorem}
\begin{proof}
Note that each $B_i$ is a well-separated set, and the number of $B_i$ is $O(\log{n})$, so it suffices to show that each $B_i$ is indeed scheduled into $O(p)$ slots w.r.t. the mean power assignment.
According to Theorem~\ref{T:uniforms},  each $Q_i$ is a $\Omega(q^\alpha)$-signal set w.r.t. uniform power, because each $K_j$ is a nearly equilength set of links, which is also $q$-independent. Using Theorem~\ref{T:robustness}, $K_j$ can be transformed into a $f$-signal schedule with at most $O((f/q^\alpha)^2)$ slots, where $f=2^{\alpha/2+1}$. Let $S_j$ be some slot from the resulting schedule of $K_j$. Let $S=\cup_j{S_j}$. For completing the proof it is enough to show that $S$ is scheduled into $p+1$ SINR-feasible slots.

 For scheduling $S$ the algorithm considers $p+1$ slots $T_i$ for $i=1,2,\dots,p+1$. The algorithm assigns each link $v$ to a slot $T_r$, which does not contain links $w$, such that $l_w\geq n^2l_v$ and $v$ and $w$ are $1/(2n)$-close. Such a set exists because the set $Q$ is $p$-bounded. Consider a link $v\in T_r$ which we took from the slot $S_k$. The affectance by the links which are nearly equilength with $v$ (i.e. links from $S_k\cap T_r$) is at most $1/f$ since the $f$-signal property holds. Changing the power assignment in the group $S_k$ from uniform to mean power increases the affectance by at most $2^{\alpha/2}$, so overall the affectance by the links with nearly the same length as $v$ is at most $2^{\alpha/2}/f=1/2$. For the links from $T_r\setminus S_k$ we have that each of them affects $v$ by less than $1/(2n)$, and since their number is at most $n$, the total affectance by those links, according to the additivity of affectance is at most $1/2$. This shows that $a_{T_r}(v)< 1$, i.e. $T_r$ is SINR-feasible, which completes the proof. 
\end{proof}

 Using the above mentioned algorithm one gets ``short'' schedules for a given $q$-independent set of links, so the next step towards solving PC-scheduling problem is to split the set $L$ into a small number of $q$-independent subsets.

\section{How Good Can the Mean Power Be for PC-scheduling?}

Note that at this point we already can prove bounds for the mean power assignments. According to Lemma~\ref{L:signalslot} a SINR-feasible set is a $1$-independent set, i.e. each schedule splits the set $L$ into $1$-independent subsets, with the number of subsets equal to the length of the schedule. The following theorem is an important result from~\cite{halldorsson} (it is stated in a slightly modified form), which states that independent sets are $p$-bounded for certain value $p$.

Let $\Lambda$ denote the ratio between the maximum and the minimum length of links.
\begin{theorem}\label{T:lambda}\cite{halldorsson}
In the case of directed scheduling each $3$-independent set of links is $p$-bounded with $p=O(\log{\log{\Lambda}})$. In the case of bidirectional scheduling each $2$-independent set of links is $1$-bounded.
\end{theorem}

\begin{theorem} \label{C:corollary1}
 For the directed model of communication the mean power assignment is a $O(\log{n}\log{\log{\Lambda}})$-approximation for the problem PC-scheduling in fading metrics.
 For bidirectional model of communication the mean power assignment is a $O(\log{n})$-approximation for the problem PC-scheduling in fading metrics.
\end{theorem}
\begin{proof}
We prove the claim for the directed model. The proof for the bidirectional model is similar.
Suppose we are given the optimal power assignment and the optimal schedule $\Sigma$ w.r.t. that power assignment. Obviously, $\Sigma$ is a $1$-signal schedule (according to our notation). Using the algorithm from Theorem~\ref{T:robustness}, $\Sigma$ can be converted to a $3^\alpha$-signal schedule $\Sigma'=(S_1,S_2,\dots,S_k)$, by increasing the length only by a constant factor. According to Lemma~\ref{L:signalslot} each $S_i$ is a $3$-independent set, so from Theorem~\ref{T:lambda} 
we have that each set $S_i$ is $p$-bounded with $p=O(\log{\log{\Lambda}})$. By applying Theorem~\ref{T:scheduleindependent}, each $S_i$ can be scheduled into $O(\log{n}\log{\log{\Lambda}})$ slots, so the whole set $L$ can be scheduled using $O(\log{n}\log{\log{\Lambda}}\cdot k)$ slots with the mean power assignment, which completes the proof. 
\end{proof}

In next two sections we present an algorithm which approaches the bounds described in Theorem~\ref{C:corollary1}.

\section{Splitting $L$ into $q$-independent Subsets}

We use graph-theoretic results for showing that a set of links can be split into a near-minimal number of $q$-independent subsets. 
First we present an algorithm for coloring a certain class of graphs, which we call \emph{$t$-strong graphs}.

\subsection{$t$-strong Graphs}

Let $G$ be a simple undirected graph. We denote by $V(G)$ the vertex-set of $G$.
  For a vertex $v$ of $G$ we denote by $N_G(v)$(or simply $N(v)$) the subgraph of $G$ induced  by the set of neighbors of $v$ in $G$.
   
   For an integer $t>0$ we say $G$ is a \emph{$t$-strong graph} if for each induced subgraph $G'$ of $G$ there is a vertex $v$ in $G'$, such that the graph $N_{G'}(v)$ does not have independent sets of size more than $t$.
   
Using the ideas of \cite{udg} for coloring Unit Disk Graphs, we prove that there is a $t$-approximation algorithm for coloring a $t$-strong graph.
The following theorem from \cite{hochbaum} describes the algorithm which we use. It is based on the results of \cite{szekeres}.
\begin{theorem}\label{T:hochbaum}
\cite{hochbaum} Let $G=(V,E)$ be a simple undirected graph and let $\delta(G)$ denote the largest $\delta$ such that $G$ contains a subgraph in which every vertex has a degree at least $\delta$. Then there is an algorithm coloring $G$  with $\delta(G)+1$ colors, with running time $O(|V|+|E|)$.
\end{theorem}

We will refer to the algorithm from Theorem \ref{T:hochbaum} as \emph{Hochbaum's algorithm}. The proof of the following theorem is similar to the proof of Theorem 4.5 of~\cite{udg} and is presented in the appendix.

\begin{theorem}\label{T:coloringdstrong}
Hochbaum's algorithm applied to a $t$-strong graph $G$ gives a $t$-approximation to the optimal coloring.
\end{theorem}
\begin{proof}
Let $OPT$ denote the number of colors used in the optimal coloring of $G$, $A$ denote the number of colors used by Hochbaum's algorithm , and $\delta(G)$ be as in Theorem \ref{T:hochbaum}.
According to Theorem \ref{T:hochbaum},
 $ A\leq \delta(G)+1$.
 
Now let $H$ be a subgraph of $G$ in which every vertex has a degree at least $\delta(G)$. According to the definition of $t$-strong graphs, there is a vertex $v$ in $H$, for which the graph $N_H(v)$ has no independent set with more than $t$ vertices, so any vertex coloring of $N_H(v)$ uses at least $|V(N_H(v))|/t$ colors. On the other hand, from the definition of $H$ we have $|V(N_H(v))|\geq \delta(G)$, so for coloring the subgraph of $G$ induced by the vertex-set $V(N_H(G))\cup \{v\}$ we need at least $\delta(G)/t+1$ colors, so
\[
OPT\geq \delta(G)/t+1\geq (A-1)/t+1,
\]
or $A\leq t\cdot OPT-t+1$, which completes the proof.
\end{proof}

\subsection{Link-graphs are $O(1)$-strong}

For $q\geq 1$, when the directed model of communication is considered, let $D_q(L)$ be the graph with vertex set $L$ (i.e. the vertices are the links from $L$), where two vertices $v$ and $w$ are adjacent in $D_q(L)$ if and only if $v$ and $w$ are \emph{not} $q$-independent w.r.t. the mean power assignment, i.e.
\begin{equation}\label{E:adjacentd}
 \mbox{either } d_{vw} \leq q \sqrt{l_w l_v} \mbox{ or } d_{wv} \leq q \sqrt{l_w l_v}.
 \end{equation}
  For the bidirectional model let $B_q(L)$ be the graph with vertex set $L$ and with two vertices $v$ and $w$ adjacent if and only if they are not $q$-independent, i.e.
\begin{equation}\label{E:adjacentb}
d_{vw} \leq q \sqrt{l_w l_v}.
\end{equation}

We show that $B_q(L)$ is $t$-strong, and $D_q(L)$ is $t'$-strong for some constants $t,t'>0$, so that Hochbaum's algorithm finds colorings for those graphs, which approximate the respective optimal colorings within constant factors. This is shown in the following theorems.

The use of the properties of doubling metrics is encapsulated in the following lemma.
\begin{lemma}\label{L:points}
Let $\{t_0,t_1,t_2,\dots,t_k\}$ be a set of points in an $m$-dimensional doubling metric space and $c_1,c_2,c_3$ and $\{b_0,b_1,b_2,\dots,b_k\}$ be positive reals, such that 

1) $b_0\leq c_1b_i$, for $i=1,2,\dots,k$,

2) $d(t_0,t_i)\leq c_2b_0b_i$ for $i=1,2,\dots,k$ and 

3) $d(t_i,t_j)>c_3b_ib_j$ for $i,j=1,2,\dots,k, i\neq j$.\\
Then $\displaystyle k\leq C(\frac{4c_2}{c_1c_3^2}+1)^m+1$.
\end{lemma}
\begin{proof}
From the triangle inequality, for $i,j=1,2,\dots,k, i\neq j$ we have 
\[
d(t_i,t_j)\leq d(t_0,t_i)+d(t_0,t_j),
\]
so using 2) for the left side and 3) for the right side, we get 
\begin{equation}\label{E:bs}
c_2b_0b_i+c_2b_0b_j>c_3b_ib_j
\end{equation}
Suppose the smallest between $b_i$ and $b_j$ is $b_i$. Then from (\ref{E:bs}) we get
$\displaystyle b_0>\frac{c_3}{2c_2}b_i$, thus we have that $b_0$ is more than $\displaystyle \frac{c_3}{2c_2}b_i$ for all $i>0$ but one: without loss of generality suppose those indices are $1,2,\dots,k-1$. Then we have
\[
d(t_0,t_i)<\frac{2c_2}{c_3}b_0^2\mbox{ and }d(t_i,t_j)>c_1c_3b_0^2
\]
for $i,j=1,2,\dots,k-1, i\neq j$. The last two inequalities imply that the balls $B(t_i,c_1c_3b_0^2/2)$ for different $i$ don't intersect, and are contained in the ball $\displaystyle B(t,(2c_2/c_3+c_1c_3/2)b_0^2)$. As the metric space has a doubling dimension $m$, we get $\displaystyle k-1\leq C(\frac{4c_2}{c_1c_3^2}+1)^m$, which completes the proof. 
\end{proof}

\begin{theorem}\label{T:strongb}
The graph $B_q(L)$ is $O(1)$-strong.
\end{theorem}
\begin{proof}
Consider the vertex $v$ with $l_v$ being minimum over all links, and a subset $I=\{1,2,\dots,k\}$ of vertices of $N(v)$, which is an independent set in $N(v)$. Our goal is to show that $|I|=O(1)$.

Consider the set of nodes $R=\{t_1,t_2,\dots,t_k\}$, where $t_i$ is the node (sender or receiver) of the link $i$, closest to the link $v$ (in terms of the distance between two sets of points). $R$ can be split into two subsets, first with nodes for which the closest node of $v$ is the sender of $v$, and the others for which the receiver of $v$ is closer. We assume that $R$ is anyone of that subsets: if we show that $|R|=O(1)$, then the proof follows. We denote by $t_0$ the node of $v$ which is closer to  $R$ than the other one.

Let us denote $b_i=\sqrt{l_{i}}$ for each link $i$, and $b_0=\sqrt{l_v}$. According to (\ref{E:adjacentb}) we have 
\begin{equation}
d(t_0,t_i)\leq qb_0b_{i}
\end{equation}
\begin{equation}
d(t_i,t_j)>qb_{i}b_{j},\mbox{ for }i,j=1,2,\dots,k,i\neq j,
\end{equation}
which means that we can apply Lemma~\ref{L:points} with points $t_0,t_1,\dots,t_k$, reals $b_0,b_1,\dots,b_k$ and $c_1=1, c_2=c_3=q$, getting
\[
|R|=k\leq C\left(4/q+1\right)^m+1,
\]
thus completing the proof. 
\end{proof}

For the case of the directed model we need the following lemma.
\begin{lemma}\label{L:directlem}
Consider the directed model. Let $I$ be an $r$-independent set of links in a doubling metric, and $v\notin I$ be a link, such that for each $w\in I$, $l_w\geq hl_v$, where  $r\geq 2$ and $h\geq 1$. Further, let $\min\{d_{vw},d_{wv}\}\leq r'\sqrt{l_vl_w}$, for $r'>0$. Then \[|I|\leq 2C\left(\frac{4r'}{h(r-1)^2}+1\right)^m+1.\] 
\end{lemma}
\begin{proof}
For simplicity of notation let us assume that $I=\{1,2,\dots,|I|\}$. Since for each different $u,w=1,2,\dots,|I|$, $u$ and $w$ are $r$-independent, then we have $d_{uw}>r\sqrt{l_{u}l_{w}}$ and $d_{wu}>r\sqrt{l_{u}l_{w}}$. Let us assume that $l_{u}\leq l_{w}$. Then using the triangle inequality we have $d(s_w,s_u)\geq d_{wu}-l_u>r\sqrt{l_ul_w}-l_u$,  and since $l_{u}\leq l_{w}$, we get
\begin{equation}\label{E:directlem1}
d(s_w,s_u)>(r-1)\sqrt{l_ul_w}.
\end{equation}
With a similar argument we get
\begin{equation}\label{E:directlem2}
d(r_w,r_u)>(r-1)\sqrt{l_ul_w}.
\end{equation}
From the condition of the lemma we have that for each $w\in I$, either $d_{vw}\leq r'\sqrt{l_vl_w}$ holds or  $d_{wv}\leq r'\sqrt{l_vl_w}$. Consider the node $t_0$ and the set of nodes $R$, which we define differently depending on the following two cases:

\emph{Case 1.} There is a subset $I_1\subseteq I$ with $|I_1|\geq |I|/2$, such that $d_{vw}\leq r'\sqrt{l_vl_w}$ for all $w\in I_1$. Then we take $t_0$ to be the sender node of $v$, i.e. $s_v$, and $R$ to be the set of receiver nodes of the links from $I_1$, i.e.  $R=\{r_w|w\in I_1\}$.

\emph{Case 2.} There is a subset $I_2\subseteq I$ with $|I_2|\geq |I|/2$, such that $d_{wv}\leq r'\sqrt{l_vl_w}$ for all $w\in I_2$. 
Then we take $t_0$ to be the receiver node of $v$, i.e. $r_v$, and $R$ to be the set of sender nodes of the links from $I_2$, i.e.  $R=\{s_w|w\in I_2\}$.

In both cases $|R|\geq |I|/2$, so next we bound $|R|$.

Consider the first case. Let $|R|=k$, and, without loss of generality, $R=\{r_1,r_2,\dots,r_k\}$. Then from the definition of $R$ and $t_0$ we have that $d(t_0,r_w)\leq r'\sqrt{l_vl_w}$ for $w=1,2,\dots,k$. On the other hand, from (\ref{E:directlem2}) we have $d(r_u,r_w)>(r-1)\sqrt{l_ul_w}$ for $u,w=1,2,\dots,k,u\neq w$. Then by denoting $b_0=\sqrt{l_v}$, $t_i=r_i$ and $b_i=\sqrt{l_i}$ for $i=1,2,\dots,k$, we have
\begin{equation}
d(t_0,t_i)\leq r'b_0b_{i}
\end{equation}
\begin{equation}
d(t_i,t_j)>(r-1)b_{i}b_{j},\mbox{ for }i,j=1,2,\dots,k,i\neq j,
\end{equation}
so we can apply Lemma~\ref{L:points} with points $t_0,t_1,\dots,t_k$, reals $b_0,b_1,\dots,b_k$ and $c_1=h, c_2=r', \displaystyle c_3=r-1$, getting
\[
|R|=k\leq C\left(\frac{4r'}{h(r-1)^2}+1\right)^m+1.
\]
For the second case the proof can be completed the same way, using (\ref{E:directlem1}).
\end{proof}

\begin{theorem}\label{T:strongd}
For $q\geq2$, the graph $D_q(L)$ is $O(1)$-strong.
\end{theorem}
\begin{proof}
 Consider the vertex $v$ with $l_v$ being minimum over all links. Suppose that $I$ is a subset of vertices of $N(v)$, which is also an independent set in $N(v)$. We have that $I$ is a $q$-independent set of links, and for each $w\in I$, $l_w\geq l_v$ and $w$ and $v$ are adjacent, so according to (\ref{E:adjacentd}), $\min\{d_{vw},d_{wv}\}\leq q\sqrt{l_vl_w}$, so applying Lemma~\ref{L:directlem}, we get
$|I|\leq 2C\left(4q/(q-1)^2+1\right)^m+1=O(1).$
\end{proof}

\section{Scheduling Using the Mean Power Assignment}

Now let us go back to the problem of  PC-scheduling in a fading metric. Consider the following algorithm for scheduling $L$. 

\begin{algorithm}\caption{Scheduling arbitrary sets of links.}\label{A:schedule}
\begin{enumerate}
\item{Construct the graph $B_2(L)$ (respectively $D_3(L)$ for the directed model)}
\item{Applying the algorithm from Theorem~\ref{T:hochbaum} on the resulting graph, split $L$ into $2$-independent ($3$-independent) subsets $S_1,S_2,\dots,S_k$}
\item{For $i=1,2,\dots,k$ apply Algorithm~\ref{A:indep} to the set $S_i$, getting a schedule $\Sigma_i=\{S_i^1,S_i^2,\dots,S_i^{k_i}\}$}
\item{Output the schedule $\cup_i{\Sigma_i}$}
\end{enumerate}
\end{algorithm}

\begin{theorem}
For the bidirectional model of communication Algorithm~\ref{A:schedule} approximates PC-scheduling within a factor $O(\log{n})$ in fading metrics. 
\end{theorem}
\begin{proof}
According to Theorem~\ref{T:robustness}, for  a constant $q\geq 1$ an optimal $q^\alpha$-signal schedule is a constant factor approximation for an optimal SINR-feasible schedule. But from Lemma~\ref{L:signalslot} we know that \emph{each} $q^\alpha$-signal schedule induces a coloring of the graph $B_q(L)$, so the chromatic number of $B_q(L)$ is not more than the length of the optimal $q^\alpha$-signal schedule. So if we denote the length of an optimal SINR-feasible schedule by $OPT$, then on the second step of the algorithm we have $k=O(OPT)$. According to Theorem~\ref{T:lambda}, on the third step of the algorithm for all $i=1,2,\dots,k$ we have $k_i=O(\log{n})$, so the length of the resulting schedule on the fourth step is
$\sum_{i=1}^k{k_i}=O(\log{n}\cdot OPT)$
for the bidirectional model.
\end{proof}

We need the following theorem, to prove an approximation factor for the directed model.
\begin{theorem}\label{T:lambdam}
In the directed model each set of links, which is $3$-independent w.r.t. the mean power assignment, is $O(1)$-bounded.
\end{theorem}
\begin{proof}
Suppose that in a fading metric $I$ is a 3-independent subset of links and $v$ is a link, such that for each $w\in I$ we have $l_w\geq n^2\cdot l_v$ and $\max\{a_v(w), a_w(v)\}\geq 1/(2n)$. Then we have for each $w\in I$, $\min\{d_{vw},d_{wv}\}\leq (2n)^{1/\alpha}\sqrt{l_vl_w}$, so by applying Lemma~\ref{L:directlem} we get that \[ I\leq 2C\left(\frac{2^{1/\alpha}}{n^{2-1/\alpha}}+1\right)^m+1=O(1).\]
\end{proof}

\begin{theorem}\label{T:meanapprox}
For the directed model of communication Algorithm~\ref{A:schedule} approximates scheduling problem with fixed mean power assignment within a factor $O(\log{n})$ in fading metrics.
\end{theorem}
\begin{proof}
It is easy to see, that for $q\geq 1$, each $q^\alpha$-signal schedule, which uses the mean power assignment, induces a coloring of the graph $D_q(L)$, so the chromatic number of $D_q(L)$ is not more than the length of the optimal $q^\alpha$-signal schedule w.r.t. the mean power assignment, so if the optimal SINR-feasible schedule length (w.r.t. the optimal power assignment) is $OPTM$, then on the second step we have $k=O(OPTM)$. On the other hand, using Theorem~\ref{T:lambdam} we have that for $i=1,2,\dots k$, $k_i=O(\log{n})$, which implies that the resulting schedule length is $O(\log{n}\cdot OPTM)$.
\end{proof}

On the other hand, from Theorem~\ref{C:corollary1} we know that the mean power assignment approximates the problem of PC-scheduling within a factor of $O(\log{n}\log{\log{\Lambda}})$, hence using the algorithm Schedule to solve the PC-scheduling problem in the directed model, and taking into account Theorem~\ref{T:meanapprox}, we get a $O(\log^2{n}\log{\log{\Lambda}})$-approximation.

\begin{corollary}
For the directed model of communication Algorithm~\ref{A:schedule} approximates PC-scheduling within a factor $O(\log^2{n}\log{\log{\Lambda}})$ in fading metrics.
\end{corollary}

\section{Introducing the noise factor}
All the results we derived are for the case when there is no ambient noise factor in SINR formula. To see how much is the impact   of introducing the noise factor into the formula on the schedule length, first let us notice that if there is a noise $N$, then for each link, which is scheduled in a SINR-feasible set w.r.t. any power assignment $P$, the following must hold:
\begin{equation*}
P_v/l_v^\alpha \geq \beta N\mbox{, for each link }v.
\end{equation*}
This is the minimum power needed to deliver a message to the receiver of $v$ even if there are no other transmissions. We assume a little stronger constraint on the power assignment, i.e.
\begin{equation}\label{E:noise}
P_v/l_v^\alpha \geq 2\beta N\mbox{, for each link }v.
\end{equation}
With this assumption we can include the noise factor into SINR formula by changing our results only by a constant factor. Here is how to do it. 
If there is a set $S$, which is SINR-feasible w.r.t. power assignment $\{P_v\}$ and without noise factor and $\beta'=2\beta$, then for each $v\in S$ we have $P_v/l_v^\alpha>2\beta \sum_{w\in S\setminus v}{P_w/d_{wv}^\alpha}$. Then using (\ref{E:noise}), we get  $P_v/l_v^\alpha>\beta \sum_{w\in S\setminus v}{P_w/d_{wv}^\alpha}+\beta N$, which, taking into account Theorem~\ref{T:robustness}, proves the following theorem.
\begin{theorem}
 If (\ref{E:noise}) holds, then each zero-noise schedule of length $T$ can be transformed into a non-zero-noise schedule of length  $O(T)$.
 \end{theorem}
 It follows that if (\ref{E:noise}) holds, then our results hold with any non-zero noise factor as well.

\section*{Acknowledgment}
Author thanks Prof. M.M. Halld\'{o}rsson for helpful discussions.

\end{document}